\pdfoutput=1
\RequirePackage{ifpdf}
\ifpdf 
\documentclass[pdftex]{sigma}
\else
\documentclass{sigma}
\fi

\numberwithin{equation}{section}

\newtheorem{Theorem}{Theorem}[section]
\newtheorem{Corollary}[Theorem]{Corollary}
\newtheorem{Proposition}[Theorem]{Proposition}
 { \theoremstyle{definition}
\newtheorem{Definition}[Theorem]{Definition}
 }

\usepackage[all]{xy}

\begin{document}


\newcommand{\arXivNumber}{1607.01626}

\renewcommand{\PaperNumber}{098}

\FirstPageHeading

\ShortArticleName{Variational Tricomplex, Global Symmetries and Conservation Laws of Gauge Systems}

\ArticleName{Variational Tricomplex, Global Symmetries\\ and Conservation Laws of Gauge Systems}

\Author{Alexey A.~SHARAPOV}

\AuthorNameForHeading{A.A.~Sharapov}

\Address{Physics Faculty, Tomsk State University, Lenin ave.~36, Tomsk 634050, Russia}
\Email{\href{mailto:sharapov@phys.tsu.ru}{sharapov@phys.tsu.ru}}

\ArticleDates{Received July 12, 2016, in f\/inal form September 30, 2016; Published online October 03, 2016}

\Abstract{Using the concept of variational tricomplex endowed with a presymplectic structure, we formulate the general notion of symmetry. We show that each generalized sym\-metry of a gauge system gives rise to a sequence of conservation laws that are represented by on-shell closed forms of various degrees. This extends the usual Noether's correspondence between global symmetries and conservation laws to the case of lower-degree conservation laws and not necessarily variational equations of motion. Finally, we equip the space of conservation laws of a given degree with a Lie bracket and establish a homomorphism of the resulting Lie algebra to the Lie algebra of global symmetries.}


\Classification{70S10; 81T70; 83C40}

\section{Introduction}

 In this paper, we continue our study of the variational tricomplex and its applications initiated in \cite{Sh1}. Loosely, the variational tricomplex may be viewed as the standard variational bicomplex \cite{Anderson,Dickey,Olver,Saunders} endowed with one more coboundary operator, namely, the classical BRST dif\/ferential. The BRST dif\/ferential carries an exhaustive information about the classical equations of motion, their gauge symmetries and identities. Although the BRST theory is commonly regarded as a tool for quantizing gauge theories \cite{HT}, the classical BRST dif\/ferential, as such, has nothing to do with quantization: to quantize a classical theory one or another extra structure is needed. In the context of the variational tricomplex such an extra ingredient is most naturally identif\/ied with a BRST-invariant presymplectic structure. Depending on the formalism one uses to describe classical dynamics, dif\/ferent kinds of objects can be identif\/ied as presymplectic structures. Within the Lagrangian formalism, for example, the presymplectic structure appears as an odd symplectic form underlying the BV bracket on the space of f\/ields and antif\/ields. In the BFV formalism of constrained Hamiltonian systems the same presymplectic structure reincarnates as the canonical symplectic structure on the extended phase space. As was shown in~\cite{Sh1}, the concept of variational tricomplex provides a uniform geometrical description of all these reincarnations, maintaining an explicit space-time covariance even in the Hamiltonian picture of dynamics. In particular, it allows one to pass directly from the BV to BFV formalism at the level of the BRST charge and master action; in so doing, the whole spectrum of BFV f\/ields and the presymplectic structure are generated immediately from those of the BV theory. Moreover, with due def\/inition of the BRST dif\/ferential \cite{KazLS,LS0} the concept of variational tricomplex extends beyond the scope of Lagrangian dynamics.

In the present paper, we focus upon the issues of global symmetries, conservation laws and interrelation between them. In mathematical terms the conservation laws are described by dif\/ferential forms on an $n$-dimensional space-time manifold. The coef\/f\/icients of these forms are assumed to be given by smooth functions of f\/ields and their derivatives and the forms are required to be closed by virtue of equations of motion, that is, on-shell. Two conservation laws are considered as equivalent if they are represented by on-shell cohomologous dif\/ferential forms. The degree of a conservation law is by def\/inition the degree of a form it is represented by. Since the $n$-forms are automatically closed it makes sense to consider on-shell closed forms of degrees less or equal to $n-1$. These constitute the so-called characteristic cohomology of the system. The ``ordinary'' conservation laws have the maximal degree $n-1$, while those of degree $< n-1$ are usually referred to as the lower-degree conservation laws\footnote{In physics it is customary to describe the conservation laws by polyvectors rather than forms. Since the passage from forms to polyvectors involves the Hodge dualization w.r.t.\ to some background metric, the higher the form-degree, the lower the polyvector-degree and vice versa. Correspondingly, the lower-degree conservation laws from the viewpoint of forms become of higher-degree in terms of polyvectors.}. The typical example of the top-degree conservation law is provided by the energy-momentum tensor of the electromagnetic f\/ield, which in actuality represents four independent conserved quantities. Here we also encounter the lower-degree conservation law represented by the Hodge dual of the strength $2$-form. The latter owes its existence to the gauge invariance of the electromagnetic potentials and expresses Gauss' law. Various results on lower-degree conservation laws, obtained by variety of techniques, can be found in \cite{BBH,BG,HKS,T,Tj,Ver,V,V+}.

The notion of conservation law is closely related to the idea of symmetry. Indeed, each top-degree conservation law of a Lagrangian system def\/ines and is def\/ined by a global symmetry of the action functional. This is the precise content of Noether's f\/irst theorem on the link between symmetries and conservation laws \cite{KS}. The nature of the lower-degree conservation laws is somewhat dif\/ferent. As is well known \cite{BBH}, they owe their origin to (the special structure of) gauge symmetries, rather than to global invariance. The presence of gauge symmetries is thus a necessary but not a suf\/f\/icient condition for the existence of lower-degree conservation laws. Due to the second Noether's theorem no ordinary conserved current corresponds to the gauge invariance of the action \cite{KS}. We have to conclude that the top- and lower-degree conservation laws are quite dif\/ferent things when viewed from the perspective of the conventional Lagrangian formalism.

This dif\/ference disappears entirely within the variational tricomplex approach, where the action functional is substituted by the classical BRST dif\/ferential and the BRST-invariant presymplectic form. The global symmetries are then naturally identif\/ied with the inf\/initesimal transformations that leave invariant either structure. We show that each symmetry, being def\/ined in such a way, gives rise to a sequence of conservation laws of decreasing degree. This allows us to treat the top- and lower-degree conservation laws on equal footing, i.e., as a manifestation of global symmetries. Furthermore, using the notion of a descendent presymplectic structure \cite{Sh1}, we are able to endow the space of conservation laws of a given degree with a Lie bracket. In top-degree, this bracket reproduces the Dickey bracket in the space of conserved currents \cite{Dickey}. By construction, the Lie algebras of conservation laws come equipped with homomorphisms to the Lie algebra of original global symmetries and one may regard these homomorphisms as an extension of the f\/irst Noether's theorem to the case of lower-degree conservation laws. For Lagrangian gauge systems in the BV-BRST formalism such a connection between higher symmetries and lower-degree conservation laws was established in \cite{BBH}.

Unif\/ication of top- and lower-degree conservation laws is not the only advantage of our approach. The chief value of the concept of variational tricomplex is that it equally well applies to non-Lagrangian theories. In general, the existence of a compatible presymplectic structure imposes less restrictions on the classical dynamics than the existence of a Lagrangian. Among recent examples of this kind let us mention the derivation of conserved currents for the non-Lagrangian equations of motion governing the dynamics of massless higher-spin f\/ields \cite{Sh2}. It should be noted that one and the same system of classical equations of motion may admit, in principle, several inequivalent presymplectic structures. Not only do these presymplectic structures lead to dif\/ferent quantizations, but they also lead to dif\/ferent links between symmetries and conservation laws in the classical theory. It might be well to point out in this connection that another generalization of the f\/irst Noether's theorem to non-Lagrangian gauge theories was proposed in \cite{KLS1,KLS2}.

The paper is organized as follows. In the next section, we review the concepts of variational tricomplex and presymplectic structure. Here we also recall the notion of a descendent gauge system~\cite{Sh1}, which is basic to our subsequent considerations. In Section~\ref{section3}, we introduce and study the notions of physical observables, (Hamiltonian) symmetries and conservation laws. Among other thing we show that each Hamiltonian symmetry originates from a physical observable and the latter gives rise to a sequence of conservation laws of various degrees. In Section~\ref{section4}, we slightly relax the def\/ining conditions for a Hamiltonian symmetry and this enables us to endow the space of conservation laws with the structure of a Lie algebra. The corresponding Lie bracket is determined by the descendent presymplectic structure. In Section~\ref{section5}, the general formalism is illustrated by three examples of physical interest: Maxwell's electrodynamics, the Chern--Simons theory, and the linearized gravity in the vierbein formalism.
Appendix~\ref{A} contains some basic facts concerning the geometry of jet bundles and the variational bicomplex.

\section{Variational tricomplex of a local gauge system}\label{section2}

Let $M$ be an $n$-dimensional space-time manifold. In modern language the classical f\/ields are just the sections of a locally trivial f\/iber bundle $\pi\colon E\rightarrow M$. The typical f\/iber $F$ of $E$ is called the \textit{target space of fields}. For trivial bundles, $E=M\times F$, the f\/ields are merely smooth mappings $\phi\colon M \rightarrow F$. For simplicity we restrict ourselves to vector bundles, in which case the space of f\/ields $\Gamma (E)$ has the structure of vector space. At the same time, to accommodate fermionic f\/ields as well as ghost f\/ields associated with gauge symmetries, we assume $\pi\colon E\rightarrow M$ to be a $\mathbb{Z}$-graded supervector bundle. This means that the typical f\/iber $F$ has the structure of a $\mathbb{Z}$-graded superspace, while the base $M$ remains an ordinary (nongraded) manifold~$M$. Following the physical tradition, we refer to the $\mathbb{Z}$-grading as the \textit{ghost number} and denote the degree of a~homogeneous object~$A$ by $\operatorname{gh} (A)$. The Grassmann parity will be denoted by~$\epsilon(A)$. The latter is responsible for the sign rule. It should be emphasized that in the presence of fermions there is no natural correlation between the Grassmann parity and the ghost number. Since throughout the paper we work exclusively in the category of $\mathbb{Z}$-graded supervector bundles, we omit the boring pref\/ixes ``super'' and ``graded'' whenever possible. For a quick introduction to the graded dif\/ferential geometry we refer the reader to~\cite{CatSch,Meh,Roy,Vor1}.

 A fundamental tenet of classical f\/ield theory is locality. Above all it implies that the dynamics of f\/ields are governed by partial dif\/ferential equations. The basic tool for a geometric approach to dif\/ferential equations is provided by the jet bundle formalism. In our case a relevant jet bundle is the bundle $\pi_\infty\colon J^\infty E\rightarrow M$ of inf\/inite jets associated with the vector bundle $\pi\colon E\rightarrow M$. The dif\/ferential forms on $J^\infty E$ carry the structure of double complex. This double complex is called the variational bicomplex because one of its dif\/ferentials coincides with the variational derivative. This leads one to a formal variational calculus that can be viewed as a geometrized version of the classical calculus of variations. The free variational bicomplex represents thus a natural kinematical basis for formulating local f\/ield theories. It is summarized in Appendix~\ref{A}, where we also explain our notation. In the recent paper~\cite{Sh1}, the concept of variational bicomplex was enhanced by introducing two more geometrical ingredients: a~classical BRST dif\/ferential and a~BRST-invariant presymplectic structure. The former brings dynamics into the free variational bicomplex by making it into a tricomplex, while the latter is responsible for quantization and, as we will show below, for establishing a correspondence between symmetries and conservation laws. Let us describe these two extra structures in more detail.

\subsection{Presymplectic structure}\label{PrSt}
By a \textit{presymplectic} $(2,m)$-form on $J^\infty E$ we understand an element $\omega\in {\Lambda}^{2,m}(J^\infty E)$ satisfying
\begin{gather*}
\delta \omega \simeq 0 .
\end{gather*}
The sign $\simeq$ means equality modulo $d$-exact forms. It might be worth to mention that the horizontal degree $m$ of the presymplectic form need not be a priori related to the dimension~$n$ of the space-time manifold~$M$. Two presymplectic forms are considered as equivalent if they dif\/fer by a $d$-exact form. In what follows, we will not distinguish between $\omega\in \Lambda^{2,m}(J^\infty E)$ and its equivalence class $[\omega]$ in the quotient space $\widetilde{\Lambda}^{2,m}(J^\infty E)=\Lambda^{2,m}(J^\infty E)/d\Lambda^{2,m-1}(J^\infty E)$, denoting both by $\omega$. According to the def\/inition above the presymplectic forms are the cocycles of the relative ``$\delta$ modulo $d$'' cohomology in vertical degree 2.

The form $\omega$ is assumed to be homogeneous, so that we can speak of an odd or even presymplectic structure of def\/inite ghost number. The triviality of the relative ``$\delta$ modulo $d$'' cohomology\footnote{Recall that we have restricted ourselves to the f\/ields associated with vector bundles, where the target space of f\/ields is contractible. For more general f\/iber bundles the triviality of the relative $\delta$-cohomology should be taken as hypothesis.} in positive vertical degree (see \cite[Section~19.3.9]{Dickey}) implies that any presymplectic $(2,m)$-form is exact, namely, there exists a homogeneous $(1,m)$-form $\theta$ such that $\omega\simeq\delta\theta$. The form $\theta$ is called a \textit{presymplectic potential} for $\omega$. The presymplectic potential is obviously not unique. If~$\theta_0$ is one of the presymplectic potentials for $\omega$, then setting $\omega_0=\delta \theta_0$ we get
\begin{gather*}
\delta \omega_0=0 ,\qquad \omega_0\simeq \omega .
\end{gather*}
In other words, any presymplectic form has a $\delta$-closed representative.

An evolutionary vector f\/ield $X$ is called \textit{Hamiltonian} with respect to $\omega$ if it preserves the presymplectic form, that is,
\begin{gather}\label{Xom}
L_X\omega\simeq 0 .
\end{gather}
Obviously, the Hamiltonian vector f\/ields form a subalgebra in the Lie algebra of all evolutionary vector f\/ields. We denote this subalgebra by $\mathfrak{X}_\omega(J^\infty E)$. Equation~(\ref{Xom}) is equivalent to
\begin{gather*}
\delta i_X\omega\simeq 0 .
\end{gather*}
Again, because of the triviality of the relative $\delta$-cohomology, we can write
\begin{gather}\label{HVF}
i_X\omega \simeq \delta H
\end{gather}
for some $H\in {\Lambda}^{0,m}(J^\infty E)$. We refer to $H$ as a \textit{Hamiltonian form} (or \textit{Hamiltonian}) associated with $X$. It is clear that equation~(\ref{HVF}) def\/ines the Hamiltonian only modulo adding to $H$ a~$d$-exact form. Therefore, two Hamiltonian forms $H$ and $H'$ will be considered as equivalent if $H\simeq H'$. By abuse of notation, we will use the same symbol $H$ to denote a particular Hamiltonian form and its equivalence class. Sometimes, to indicate the relation between the Hamiltonian vector f\/ields and forms, we will write $X_H$ for $X$. In general, this relationship is far from being one-to-one.

The space ${\Lambda}_\omega^{m}(J^\infty E)$ of all Hamiltonian $m$-forms can be endowed with the structure of a Lie algebra. The corresponding Lie bracket is def\/ined as follows: If $X_A$ and $X_B$ are two Hamiltonian vector f\/ields associated with the Hamiltonian forms $A$ and $B$, then
\begin{gather}\label{PB}
\{A,B\}=(-1)^{\epsilon(X_A)}i_{X_A}i_{X_B}\omega .
\end{gather}
The next proposition shows that the bracket is well def\/ined and possesses all the required properties.

\begin{Proposition}[\cite{Sh1}]\label{1.1}
The bracket \eqref{PB} is bilinear over reals, maps the Hamiltonian forms to Hamiltonian ones, enjoys the symmetry property
\begin{gather*}
\{A,B\}\simeq -(-1)^{(\epsilon(A)+\epsilon(\omega))(\epsilon(B)+\epsilon(\omega))}\{B,A\} ,
\end{gather*}
and obeys the Jacobi identity
\begin{gather*}
\{C,\{A,B\}\}\simeq \{\{C,A\},B\}+(-1)^{(\epsilon(C)+\epsilon(\omega))(\epsilon(A)+\epsilon(\omega))}\{A,\{C,B\}\} .
\end{gather*}
\end{Proposition}
Combining equations~(\ref{HVF}) and~(\ref{PB}), one can see that
\begin{gather}\label{PB1}
\{A,B\}\simeq (-1)^{\epsilon(A)}L_{X_A} B .
\end{gather}
The last relation gives an equivalent def\/inition for the Poison bracket.

Let $\operatorname{ker}\omega$ denote the space of all Hamiltonian vector f\/ields $X$ with zero Hamiltonian, i.e.,
\begin{gather*}
i_X\omega \simeq 0 .
\end{gather*}
It is easy to see that $\operatorname{ker} \omega$ is an ideal in the Lie algebra of Hamiltonian vector f\/ields $\mathfrak{X}_\omega(J^\infty E)$. One can regard the quotient $\mathfrak{X}_\omega(J^\infty E)/\operatorname{ker} \omega$ as the Lie algebra of \textit{nontrivial} Hamiltonian vector f\/ields.
The next proposition relates this Lie algebra to the Lie algebra of Hamiltonian forms.

\begin{Proposition}\label{ShES}
There is a short exact sequence
\begin{gather*}
\xymatrix@C=0.5cm{
 0 \ar[r] &\Lambda^m(M) \ar[rr]^{\pi_\infty^\ast} &&\Lambda^{m}_\omega (J^\infty E)\ar[rr]^{\alpha} &&\mathfrak{X}_\omega(J^\infty E)/\operatorname{ker} \omega \ar[r] & 0 } ,
\end{gather*}
where $\pi_\infty^\ast$ is the pull back of the canonical projection $\pi_\infty\colon J^\infty E\rightarrow M$ and the map $\alpha$ assigns to each Hamiltonian form $A$ the equivalence class $X_A+\operatorname{ker} \omega$.
\end{Proposition}
 We leave it to the reader to check exactness. It is signif\/icant that $\alpha$ is a homomorphism of Lie algebras \cite{Sh1}. This means that
\begin{gather}\label{LAHom}
X_{\{A,B\}}=[X_A,X_B]\quad (\mathrm{mod}\quad \operatorname{ker} \omega)\qquad \forall\, A,B\in \Lambda_\omega^{m}(J^\infty E) ,
\end{gather}
and the ideal $\operatorname{ker} \alpha$ consists of the f\/ield-independent dif\/ferential forms.

It follows from the def\/inition (\ref{HVF}) that each Hamiltonian form is necessarily invariant w.r.t.\ the action of the kernel distribution, that is,
\begin{gather*}
L_XA\simeq 0\qquad \forall\, X\in \operatorname{ker} \omega ,\quad\forall\, A\in \Lambda^{m}_\omega(J^\infty E) .
\end{gather*}
Therefore, the more degenerate the presymplectic structure, the less the size of the space of Hamiltonian forms. A presymplectic form $\omega$ is called nondegenerate if $\operatorname{ker} \omega=0$, in which case we refer to it as a \textit{symplectic form}.

For a general discussion of a presymplectic structure as well as numerous applications of this notion in f\/ield theory we refer the reader to the papers \cite{AG,BHL,CW,Gr,Kh,Kh0,Sh2,Z} and the references therein.

\subsection{Classical BRST dif\/ferential}\label{brst}
An odd evolutionary vector f\/ield $Q$ on $J^\infty E$ is called \textit{homological} if
\begin{gather}\label{QQ}
[Q,Q]=2Q^2=0 , \qquad \operatorname{gh} (Q)=1 .
\end{gather}
We will use the special notation $\delta_Q$ for the Lie derivative along the homological vector f\/ield $Q$. It follows from the def\/inition that $\delta_Q^2=0$. Hence, $\delta_Q$ is a dif\/ferential of the algebra $\Lambda^{\ast,\ast}(J^\infty E)$ increasing the ghost number by 1. Moreover, the operator $\delta_Q$ anticommutes with the coboundary operators $d$ and $\delta$:
\begin{gather*}
\delta_Q d+d\delta_Q=0 ,\qquad \delta_Q \delta+\delta\delta_Q=0 .
\end{gather*}
This allows us to speak of the tricomplex $\Lambda^{\ast,\ast,\ast}(J^\infty E; \delta, d, \delta_Q)$, where
\begin{gather*}
\delta_Q\colon \ \Lambda^{p,q,r}(J^\infty E)\rightarrow \Lambda^{p,q,r+1}(J^\infty E)
\end{gather*}
and $r$ is the ghost number.

In the physical literature the operator $\delta_Q$ is called the \textit{classical BRST differential} and we will also use this term to refer to the homological vector f\/ield $Q$ itself.

The equations of motion of a gauge theory are recovered by considering the zero locus of~$Q$. In terms of the adapted coordinates $(x^i, \phi^a_I)$ on $J^\infty E$ the vector f\/ield $Q$, being evolutionary, assumes the form\footnote{We use the multi-index notation according to which the multi-index $I=i_1i_2\cdots i_k$ represents the set of symmetric covariant indices and $\partial_I=\partial_{i_1}\cdots\partial_{i_k}$. The \textit{order} of the multi-index is given by $|I|=k$.}
\begin{gather*}
Q=\partial_I Q^a\frac{\partial}{\partial \phi_I^a} .
\end{gather*}
Then there exists an integer $l$ such that the equations
\begin{gather*}
\partial_I Q^a=0 ,\qquad |I|=k ,
\end{gather*}
def\/ine a submanifold $\Sigma^k\subset J^{l+k}E$. The standard regularity condition implies that $\Sigma^{k+1}$ f\/ibers over $\Sigma^k$ for each $k$. This gives the inf\/inite sequence of projections
\begin{gather*}
\xymatrix{\cdots\ar[r]& \Sigma^{l+3}\ar[r]&\Sigma^{l+2}\ar[r]&\Sigma^{l+1}\ar[r]&{\Sigma^l}\ar[r]& M ,}
\end{gather*}
which enables us to def\/ine the zero locus of $Q$ as the inverse limit
\begin{gather*}
 \Sigma^\infty =\lim_{\longleftarrow}\Sigma^k .
 \end{gather*}
In physics, the submanifold $\Sigma^\infty\subset J^\infty E$ is usually referred to as the \textit{shell}. The terminology is justif\/ied by the fact that the classical f\/ield equations as well as their dif\/ferential consequences can be written as\footnote{To avoid any confusion, let us stress that the collection of f\/ields $\phi$ includes both the ``usual f\/ields'' (i.e., those with ghost number zero) and the ghost f\/ields. Accordingly, by the classical f\/ield equations we mean partial dif\/ferential equations for the whole collection of f\/ields $\phi$. The equations for ``usual f\/ields'' are then obtained by projecting the shell to the sector of ghost number zero.}
\begin{gather}\label{EoM}
(j^{\infty}\phi)^\ast (\partial_I Q^a)=0 .
\end{gather}
In other words, the f\/ield $\phi\in \Gamma(E)$ satisf\/ies the classical equations of motion if\/f $j^\infty \phi \in \Sigma^\infty$. In the conventional BRST theory of variational gauge systems, the relationship between the zero locus of the classical BRST dif\/ferential and solutions to the classical equations of motion was studied in \cite{GST}. The extension to non-Lagrangian gauge systems may be found in~\cite{KazLS}.

It follows from (\ref{QQ}) that the shell $\Sigma^\infty$ is invariant under the action of $Q$. This makes possible to pull the ``free'' variational tricomplex $\Lambda^{\ast,\ast,\ast}(J^\infty E; \delta,d, \delta_Q)$ back to $\Sigma^\infty$ and so def\/ine the \textit{on-shell tricomplex} $\Lambda^{\ast,\ast,\ast}(\Sigma^\infty; \delta, d, \delta_Q)$. The latter is not generally $d$-exact even locally and this gives rise to various interesting cohomology groups associated with gauge dynamics. For example, the groups $H^{0,\ast,0}(\Sigma^\infty; d)$ describe the so-called \textit{characteristic cohomology} of a gauge system, see \cite{BBH,BG,HKS,T,Tj,Ver,V,V+} and Section~\ref{section3} below. The interpretation of some other groups can be found in~\cite{KLS2}.

It should be noted that the f\/irst variational tricomplex for gauge systems was introduced in \cite{BH} as the Koszul--Tate resolution of the usual variational bicomplex for partial dif\/ferential equations. Using this tricomplex, the authors of \cite{BH} were able to relate various Lie algebras associated with the global symmetries and conservation laws of a classical gauge system. Our tricomplex is similar in nature but involves the full BRST dif\/ferential, and not its Koszul--Tate part.

\subsection[$Q$-invariant presymplectic structure and its descendants]{$\boldsymbol{Q}$-invariant presymplectic structure and its descendants}
By a \textit{gauge system} on $J^\infty E$ we mean a pair $(Q, \omega)$ consisting of a classical BRST dif\/ferential $Q$ and a $Q$-invariant presymplectic $(2,m)$-form $\omega$. In other words, the vector f\/ield $Q$ is supposed to be Hamiltonian with respect to $\omega$, so that $\delta_Q\omega\simeq 0$. Then, according to (\ref{Xom}) and (\ref{HVF}), there exist forms $\omega_1$, $L$, and $\theta_1$ such that
\begin{gather}\label{des}
\delta_Q \omega=d\omega_1 , \qquad i_Q\omega =\delta L +d\theta_1 ,
\end{gather}
with $L$ being the Hamiltonian for $Q$ relative to $\omega$. As was mentioned in Section~\ref{PrSt}, we can always assume that $\omega =\delta\theta$ for some presymplectic potential $\theta$, so that $\delta\omega=0$. Then applying~$\delta$ to the second equality in (\ref{des}) and using the f\/irst one, we f\/ind $d(\omega_1-\delta\theta_1)=0$. On account of the exactness of the variational bicomplex the last relation is equivalent to
\begin{gather*}
\omega_1\simeq \delta\theta_1 .
\end{gather*}
Thus, $\omega_1$ is a presymplectic $(2,m-1)$-form on $J^\infty E$ with the presymplectic potential $\theta_1$. Furthermore, the form~$\omega_1$ is $Q$-invariant as one can easily see by applying $\delta_Q$ to the f\/irst equality in~(\ref{des}) and using once again the fact of exactness of the variational bicomplex. Let $L_1$ denote the Hamiltonian for $Q$ with respect to~$\omega_1$, i.e.,
\begin{gather*}
i_Q\omega_1\simeq \delta L_1 , \qquad L_1\in {\Lambda}^{0,m-1}(J^\infty E) .
\end{gather*}
Given the pair $(Q,\omega)$, we call $\omega_1$ the \textit{descendent presymplectic structure} on $J^\infty E$ and refer to $(Q,\omega_1)$ as the \textit{descendent gauge system}. This construction of a descendent gauge system can be iterated producing a sequence of gauge systems $(Q, \omega_k)$, where the $k$-th presymplectic form $\omega_k\in {\Lambda}^{2,m-k}(J^\infty E)$ is the descendant of the previous form $\omega_{k-1}$. The minimal $k$ for which $\omega_{k} \simeq 0$ gives a numerical invariant of the original gauge system $(Q,\omega)$. We call it the \textit{length of a gauge system}.

\section{Symmetries, observables and conservation laws}\label{section3}

\begin{Definition}\label{DCL} Given a classical BRST dif\/ferential $Q$, a form $\alpha \in \Lambda^{0,m}(J^\infty E)$ is said to def\/ine a conservation law of degree $m$ if
\begin{gather}\label{CL}
d\alpha|_{{\Sigma^\infty}}= 0 .
\end{gather}
The conservation law is called trivial if $\alpha|_{\Sigma^\infty}\simeq 0$.
\end{Definition}
In other words, the conservation laws are represented by the on-shell closed forms and the trivial conservation laws correspond to the on-shell exact forms. This allows us to identify the space of nontrivial conservation laws with the cohomology groups $H^{0,m}(\Sigma^\infty; d)$ of the on-shell variational bicomplex. In addition to the form degree these groups are also graded by the ghost number. In what follows, the form degree of a horizontal form $\alpha\in \Lambda^{0,m}(J^\infty E)$ will be denoted by $\deg \alpha =m$.

Due to the standard regularity condition \cite[Section~5.1]{BBH}, equation~(\ref{CL}) implies the existence of a form $\chi$ such that
\begin{gather*}
d\alpha=i_Q\chi .
\end{gather*}
The form $\chi$ is called the \textit{characteristic} of the conservation law $\alpha$. Note that shifting a characteristic by a $d$-exact form one does not change the equivalence class of the corresponding conservation law. This gives a natural equivalence relation on the space of characteristics. A~characteris\-tic~$\chi$ is called trivial if $\chi\simeq 0$.

Given a conservation law represented by an $m$-form $\alpha$ together with an $m$-cycle $C\subset M$ and a f\/ield conf\/iguration $\phi\in \Gamma(E)$, we can def\/ine the integral
\begin{gather*}
I[\phi]=\int_C (j^{\infty}\phi)^\ast(\alpha) .
\end{gather*}
By construction, the integrand is given by a closed form on $M$ provided that $j^\infty\phi\in \Sigma^\infty$. Therefore, for a f\/ixed solution $\phi$, the value of the integral depends only on the homology class of $C$ in $M$. It is the invariance of the functional $I[\phi]$ under continues deformations of $C$ which is usually meant by a conservation law\footnote{In physical problems the $m$-chain $C$ is often noncompact (e.g., a~time-slice in the Minkowski space), in which case some appropriate asymptotic conditions on the f\/ields are imposed to provide the existence and conservation of the charge $I[\phi]$.}. The functional $I[\phi]$ is called the \textit{conserved charge}.

\begin{Definition}\label{observable}
A form $\alpha\in \Lambda^{0,m}(J^\infty E)$ is called an observable of degree $m$ if
\begin{gather*}\delta_Q\alpha\simeq 0 .\end{gather*}
An observable $\alpha$ is said to be trivial if $\alpha\simeq\delta_Q\beta$ for some $\beta$.
\end{Definition}
According to this def\/inition the space of nontrivial observables of degree $m$ and ghost num\-ber~$r$ is identif\/ied with the cohomology groups $H^{0,m,r}(J^\infty E;\delta_Q)$. (Here we slightly deviate from the standard usage. Usually, by an observable in the BRST theory one means a $Q$-invariant quantity with ghost number zero, which corresponds to a gauge invariant local observable. According to our def\/inition an observable may have nonzero ghost number.)

\begin{Proposition}\label{DesCL} Suppose that the complex
\begin{gather}\label{exact}
\xymatrix{& 0\ar[r]&\mathbb{R}\ar[r] &\Lambda^{0,0}(J^\infty E) \ar[r]^-{d} & \Lambda^{0,1}(J^\infty E) \ar[r]^-{d} & \cdots \ar[r]^-{d} & \Lambda^{0,n}(J^\infty E) \\}
\end{gather}
is exact\/\footnote{This is the case, for example, when $M\simeq\mathbb{R}^n$.}. Then each observable $\alpha_0$ of degree $m $ gives rise to the sequence of observables and conservation laws $\{\alpha_k\}_{k=1}^{m}$, where the characteristic of $\alpha_{k}$ is given by $\delta\alpha_{k-1}$ and $\deg \alpha_k=m-k$. Trivial observables give rise to trivial conservation laws.
\end{Proposition}
Note that the proposition does not assert that all conservation laws originating from a nontrivial observable are nontrivial.

\begin{proof} We use the cohomological descent method \cite{BBH}. From Def\/inition \ref{observable} of an observable it follows that
\begin{gather}\label{da1}
\delta_Q\alpha_0=d\alpha_1
\end{gather}
for some $\alpha_1$ of degree $m-1$. By Def\/inition \ref{DCL}, $\alpha_1$ is a conservation law with characteristic $\delta\alpha_0$. Applying the dif\/ferential $\delta_Q$ to both sides of (\ref{da1}) yields $d\delta_Q\alpha_1=0$. The complex (\ref{exact}) being exact, we can write $\delta_Q\alpha_1=d\alpha_2$ for some $\alpha_2\in \Lambda^{0,m-2}(J^\infty E)$. Thus, $\alpha_1$ is an observable and $\alpha_2$ is a conservation law with characteristic $\delta\alpha_1$. Iterating this construction once and again, we get the sequence $\{\alpha_k\}$ of observables and conservation laws.

If $\alpha_0$ is a trivial observable, then $\alpha_0=\delta_Q\beta +d\gamma$ and $\alpha_1=\delta_Q\gamma+d\sigma$ for some $\sigma$. Hence, $\alpha_1$~is trivial as an observable and as a conservation law.
\end{proof}

\begin{Definition}
An evolutionary vector f\/ield $X$ is called a symmetry of a gauge system if it preserves the classical BRST dif\/ferential, that is,
\begin{gather*}
[X,Q]=0 .
\end{gather*}
\end{Definition}

It follows from the def\/inition that the f\/low generated by $X$ preserves the shell $\Sigma^\infty$ mapping solutions to solutions.

\begin{Definition}
A symmetry $X$ is called trivial or a gauge symmetry, if there exists another evolutionary vector f\/ield $Y$ such that $X=[Q,Y]$.
\end{Definition}
It is easy to see that the gauge symmetries form an ideal $\mathfrak{X}_{\rm GS}(J^\infty E)$ in the Lie algebra of all symmetries $\mathfrak{X}_S(J^\infty E)$. Therefore, it is natural to identify the Lie algebra of nontrivial symmetries with the quotient $\mathfrak{X}_S(J^\infty E)/\mathfrak{X}_{\rm GS}(J^\infty E)$. The latter can also be regarded as the group of $\delta_Q$-cohomology, with the dif\/ferential $\delta_Q$ -- the Lie derivative along~$Q$ -- acting in the space of evolutionary vector f\/ields.

In this paper, we are mostly interested in the Hamiltonian symmetries of gauge systems.

\begin{Definition}\label{HSym}
A symmetry $X$ is called Hamiltonian if $X$ is a Hamiltonian vector f\/ield.
\end{Definition}

\begin{Proposition}\label{27}
The Hamiltonian of a Hamiltonian symmetry is an observable. Trivial Hamiltonian symmetries corresponds to trivial observables.
\end{Proposition}

\begin{proof}
By def\/inition we have
\begin{gather*}
i_X\omega\simeq \delta \alpha ,
\end{gather*}
where $\alpha$ is a Hamiltonian of $X$. Acting by $\delta_Q$ on both
the sides of the last expression, we get
\begin{gather*}
\delta\delta_Q\alpha\simeq 0 .
\end{gather*}
By Proposition \ref{A1},
\begin{gather}\label{AAA}
\delta_Q\alpha=\pi^\ast_{\infty}(\beta) +d\gamma ,
\end{gather}
where $\beta$ is a dif\/ferential form on $M$. If $\operatorname{gh} (\beta)=\operatorname{gh}(\alpha)+1\neq 0$, then automatically $\beta=0$ as we have no parameters with nonzero ghost number. In the general case, consider a solution $\phi\in \Gamma(E)$ to the equations of motion (\ref{EoM}). We have $(j^\infty \phi)^\ast(\delta_Q\alpha)=(j^\infty \phi)^\ast(i_Q\delta \phi)=0$. Applying now the pullback $(j^\infty\phi)^\ast$ to both the sides of
(\ref{AAA}), we f\/ind $\beta = - d(j^\infty \phi)^\ast(\gamma)$. Hence, $\delta_Q\alpha\simeq 0$ and the form $\alpha$ is an observable.

If $X=\delta_Q Y$, then, according to (\ref{PB1}) and (\ref{LAHom}), the Hamiltonian of $X$ is given by the form $\alpha=-\delta_Q\beta+d\gamma$, where $\beta$ is the Hamiltonian of $Y$ and $\gamma$ is an arbitrary $(m-1)$-form. Thus, $\alpha$~is a trivial observable.
\end{proof}

Combining the last proposition with Proposition \ref{DesCL}, we arrive at the following statement.

\begin{Corollary}\label{C28} If the sequence \eqref{exact} is exact, then each Hamiltonian symmetry gives rise to a~sequence of conservation laws, perhaps trivial.
\end{Corollary}

For example, the classical BRST dif\/ferential $Q$ can be viewed as a symmetry for itself. So, it gives rise to a conservation law $L_1$ def\/ined by the equation $\delta_Q L=dL_1$, where $L$ is the Hamiltonian of $Q$ relative to $\omega$. It is not hard to see \cite{Sh1} that the form $L'_1=L_1+i_Q\theta_1$, def\/ining an equivalent conservation law, is Hamiltonian relative to the descendent presymplectic structure $\omega_1=\delta\theta_1$. Indeed, applying $\delta_Q$ to both the sides of the second equality in (\ref{des}), we get
\begin{gather*}
i_Q\delta_Q\omega=-\delta\delta_Q L -d\delta_Q\theta_1 ,
\\
i_Qd\omega_1=-\delta dL_1-d(i_Q\delta -\delta i_Q)\theta_1 ,
\\
di_Q\omega_1=d\delta (L_1+i_Q\theta_1)-di_Q\omega_1 ,
\end{gather*}
and hence
\begin{gather*}
2i_Q\omega_1\simeq \delta (L_1+i_Q\theta_1)=\delta L'_1 .
\end{gather*}
 It then follows from equation~(\ref{LAHom}) that $\delta \{L'_1, L'_1\}\simeq 8i_{Q^2}\omega_1=0$. If $\operatorname{gh}\{L'_1,L'_1\}\neq 0$, then
\begin{gather}\label{cmeq}
{\{L'_1,L'_1\}}_{1}\simeq 0 .
\end{gather}
In the case where $L$ is a form of top horizontal degree, the integral
\begin{gather*}
\Omega[\phi]=\int_N(j^{\infty}\phi)^\ast(L'_1)
\end{gather*}
over a Cauchy hypersurface $N\subset M$ is called the \textit{classical BRST charge} and equation~(\ref{cmeq}) is known as the \textit{classical master equation}, see \cite{Sh1}.

In a sense the example of the BRST symmetry $Q$ is the exception rather than the rule. Ge\-ne\-rally the conservation laws associated with Hamiltonian symmetries are neither Hamiltonian nor equivalent to Hamiltonian (relative to the descendent presymplectic structure). In the next section, we will see that the descendent presymplectic forms do induce appropriate Lie brackets on the conservation laws of various degrees providing one properly extends the notion of a~Hamiltonian form.

\section{The Lie algebra of conservation laws}\label{section4}
As was mentioned in Section~\ref{brst} the variational tricomplex admits a consistent restriction to the shell $\Sigma^\infty$. The cochains of the on-shell tricomplex can be identif\/ied with the equivalence classes of dif\/ferential forms on $J^\infty E$, where two forms $\alpha$ and $\beta$ are considered equivalent if
\begin{gather*}
\alpha|_{\Sigma^\infty}=\beta|_{\Sigma^\infty} .
\end{gather*}
For the further convenience we also introduce the sign of ``weak equality'' $\approx$ meaning that
\begin{gather*}
\alpha\approx \beta \quad \Leftrightarrow\quad \alpha|_{\Sigma^\infty}\simeq \beta|_{\Sigma^\infty} .
\end{gather*}
Due to the regularity condition for $\Sigma^\infty$, the equation $\alpha \approx 0$ simply means that there exists a~$d$-exact form $d\sigma$ such that the dif\/ference $\alpha-d\sigma$ belongs to the dif\/ferential ideal of $\Lambda^{\ast,\ast} (J^\infty E)$ algebraically generated by all the dif\/ferential forms of the form $i_Q\beta$ and $\delta_Q \gamma$.

\begin{Definition}\label{OHVF}
A symmetry $X$ of a gauge system $(Q,\omega)$ is called on-shell Hamiltonian
if there exists a form $\alpha$ such that
\begin{gather}\label{onshham}
i_X\omega\approx \delta \alpha .
\end{gather}
\end{Definition}
 As a consequence of the def\/inition, $L_X\omega \approx 0$ for any on-shell Hamiltonian symmetry $X$. The converse is not always true as the on-shell bicomplex may not be globally exact in columns even if the underlying f\/iber bundle of f\/ields $\pi\colon E\rightarrow M$ is a vector bundle. It is obvious that the Hamiltonian symmetries form a subalgebra in the Lie algebra of all on-shell Hamiltonian symmetries. We denote the latter by ${\mathfrak{X}}_{\omega,Q}(J^\infty E)$.

 Equation~(\ref{onshham}) def\/ines $\alpha$ only modulo $d$-exact and on-shell vanishing forms. A form $\alpha$ satisfying (\ref{onshham}) for some symmetry $X$ will be called \textit{on-shell Hamiltonian}. Two on-shell Hamiltonian forms $\alpha$ and $\alpha'$ associated with one and the same symmetry $X$ will be considered as equivalent if $\alpha'\approx \alpha$. Due to the regularity of the shell the last equality is equivalent to the existence of forms $\beta$ and $\gamma$ such that $\alpha'-\alpha=i_Q\beta+d\gamma$.

\begin{Proposition}
The equivalence classes of on-shell Hamiltonian forms make a Lie algebra with respect to the bracket
\begin{gather}\label{WHF}
\{\alpha, \beta\}=(-1)^{\epsilon(X)}i_Xi_Y\omega ,
\end{gather}
where $X$ and $Y$ are symmetries associated with $\alpha$ and $\beta$, respectively.
\end{Proposition}

The proof of this proposition literally repeats that of Proposition \ref{1.1} if one replaces the equality $\simeq$ by the weaker one $\approx$. The Lie algebra of all on-shell Hamiltonian $m$-forms will be denoted by $\Lambda^m_{\omega,Q}(J^\infty E)$. For Lagrangian theories without gauge symmetries the Lie bracket~(\ref{WHF}) of the f\/irst integrals of motion was studied in \cite{Dickey}.

Let $\operatorname{ker}_Q\omega$ denote the space of all symmetries satisfying the homogeneous equation \begin{gather*}i_X\omega\approx 0 .\end{gather*}
It is clear that $\operatorname{ker}_Q\omega $ contains the intersection $\mathfrak{X}_S(J^\infty E)\cap \operatorname{ker} \omega$ and is contained in $\mathfrak{X}_{\omega,Q}(J^\infty E)$. If now $X\in \operatorname{ker}_Q\omega$ and $Y\in \mathfrak{X}_{\omega,Q}(J^\infty E)$, then
\begin{gather*}
L_Yi_X\omega\approx i_{[Y,X]}\omega\approx 0.
\end{gather*}
This shows that $\operatorname{ker}_Q\omega$ is an ideal of the Lie algebra $\mathfrak{X}_{\omega,Q}(J^\infty E)$ and we can def\/ine the quotient algebra $\mathfrak{X}_{\omega,Q}(J^\infty E)/\operatorname{ker}_Q\omega$.

\begin{Proposition}\label{ShES2}
The correspondence $ \alpha \mapsto X_\alpha+\operatorname{ker}_Q\omega$ defines a homomorphism \begin{gather*}f\colon \ \Lambda^m_{\omega,Q}(J^\infty E)\rightarrow\mathfrak{X}_{\omega,Q}(J^\infty E)/\operatorname{ker}_Q\omega\end{gather*} of the Lie algebras.
 \end{Proposition}

The proof is straightforward. Notice that $\operatorname{ker} f$ contains the f\/ield-independent forms, i.e., the elements of $\operatorname{im} \pi_\infty^\ast$. Belonging to the center of the Lie algebra $\Lambda^m_{\omega,Q}(J^\infty E)$, these forms are responsible for appearance of possible central charges in the Lie algebra of on-shell Hamiltonian symmetries or, more properly, in its preimage in $\Lambda^m_{\omega,Q}(J^\infty E)$.

\begin{Theorem}
Let $\{\alpha_k\}$ be the sequence of conservation laws associated with a Hamiltonian symmetry $X$. Then the form $\alpha_k$ is on-shell Hamiltonian w.r.t.\ the $k$-th descendent presymplectic structure.
\end{Theorem}

\begin{proof}
Let $\alpha$ be a Hamiltonian of $X$. Then
\begin{gather*}
i_X\omega = \delta\alpha +d\alpha'
\end{gather*}
form some $\alpha'$. Applying $\delta_Q$ to the last equality, we get
\begin{gather}
-(-1)^{\epsilon{(X)}}i_X\delta_Q\omega=-(-1)^{\epsilon{(X)}}\delta\delta_Q\alpha-d\delta_Q\alpha' ,\nonumber\\
-(-1)^{\epsilon{(X)}}i_Xd\omega_1=-(-1)^{\epsilon{(X)}}\delta d\alpha_1-d\delta_Q\alpha' ,\label{Xo}\\
di_X\omega_1=d\delta\alpha_1+d\delta_Q\alpha' .\nonumber
\end{gather}
This implies
\begin{gather}\label{Xo1}
i_X\omega_1=\delta\alpha_1+\delta_Q\alpha'+d\alpha'_1
\end{gather}
for some $\alpha'_1$. Hence,
\begin{gather*}
i_X\omega_1\approx\delta \alpha_1
\end{gather*}
and $\alpha_1$ is an on-shell Hamiltonian form relative to $\omega_1$. Now acting on both the sides of~(\ref{Xo1}) by~$\delta_Q$, we get
\begin{gather*}
-(-1)^{\epsilon{(X)}}i_X\delta_Q\omega_1=-(-1)^{\epsilon{(X)}}\delta \delta_Q\alpha_1-d\delta_Q\alpha'_1 .
\end{gather*}
This relation coincides in form with the f\/irst line of (\ref{Xo}). Therefore, there exists a form $\alpha'_2$ such that
\begin{gather*}
i_X\omega_2=\delta\alpha_2+\delta_Q\alpha'_1+d\alpha'_2
\end{gather*}
and we conclude that $\alpha_2$ is on-shell Hamiltonian. Iterating this construction once and again, we obtain the sequence of relations
\begin{gather}\label{oa}
i_X\omega_k=\delta\alpha_k+\delta_Q\alpha'_{k-1}+d\alpha'_k
\end{gather}
meaning that all the forms $\alpha_k$ are on-shell Hamiltonian.
\end{proof}

Combining the above theorem with Proposition \ref{ShES2}, we arrive at

\begin{Corollary}
The descendent conservation laws associated with Hamiltonian symmetries form Lie algebras w.r.t.\ the descendent Lie brackets.
\end{Corollary}

In particular, if $\operatorname{ker}_Q\omega_k=0$, then the algebra $\Lambda^{m-k}_{Q,\omega_k}(J^\infty E)$ is given by a central extension of the Lie algebra $\mathfrak{X}_{\omega,Q}(J^\infty E)$.
This statement may be viewed as a main result of the paper. It relates the conservation laws of various degrees to the symmetries of the gauge system $(Q,\omega)$, that is, to the evolutionary vector f\/ields that preserve both the classical BRST dif\/ferential $Q$ and the BRST invariant presymplectic structure $\omega$.

Given the sequence of conservation laws $\{\alpha_k\}$ associated with a Hamiltonian symmetry $X$, the minimal $k$ for which $\alpha_k\approx 0$ is called the \textit{length of the symmetry} $X$.

\section{Some applications}\label{section5}

In this section, we illustrate the general formalism above by a few examples of physical interest. Since the gauge theories we are going to consider are originally formulated in terms of the Batalin--Vilkovisky formalism, we start with a brief explanation of how this formalism f\/its into our framework. For more details we refer the reader to \cite{Sh1}.

\subsection{BV formalism} The starting point of the BV formalism is an inf\/inite-dimensional manifold $\mathcal{M}_0$ of gauge f\/ields that live on an $n$-dimensional space-time $M$. Depending on a particular structure of gauge symmetry the manifold $\mathcal{M}_0$ is extended to an $\mathbb{N}$-graded manifold $\mathcal{M}$ containing $\mathcal{M}_0$ as its body. The new f\/ields of positive $\mathbb{N}$-degree are called the \textit{ghosts} and the $\mathbb{N}$-grading is referred to as the \textit{ghost number}. Let us collectively denote all the original f\/ields and ghosts by $\Phi^A$ and refer to them as f\/ields. At the next step the space of f\/ields $\mathcal{M}$ is further extended by introducing the odd cotangent bundle $\Pi T^\ast[-1]\mathcal{M}$. The f\/iber coordinates, called \textit{antifields}, are denoted by $\Phi_A^\ast$. These are assigned with the following ghost numbers and Grassmann parities:
\begin{gather*}
\operatorname{gh} (\Phi^\ast_A)=-\operatorname{gh} \big(\Phi^A\big)-1 ,\qquad \epsilon (\Phi^\ast_A)=\epsilon \big(\Phi^A\big)+1 \quad (\mbox{mod}\, 2) .
\end{gather*}
Thus, the total space of the odd cotangent bundle $\Pi T^\ast[-1]\mathcal{M}$ becomes a $\mathbb{Z}$-graded supermanifold. The canonical Poisson structure on $\Pi T^\ast[-1]\mathcal{M}$ is determined by the following odd Poisson bracket in the space of smooth functionals of $\Phi$ and $\Phi^\ast$:
\begin{gather*}
(A,B)=\int_M \left(\frac{\delta_r A}{\delta \Phi^A}\frac{\delta_l B}{\delta \Phi^\ast_A}-\frac{\delta_r A}{\delta \Phi^\ast_A}\frac{\delta_l B}{\delta \Phi^A}\right)d^nx .
\end{gather*}
 Here $d^nx$ is a volume form on $M$ and the subscripts $l$ and $r$ refer to the standard left and right functional derivatives. In the physical literature the above bracket is usually called the \textit{antibracket} or the \textit{BV bracket}.

The functionals of the form
\begin{gather*}
A=\int_M (j^\infty \phi)^\ast(a) ,
\end{gather*}
where $\phi=(\Phi, \Phi^\ast)$ and $a\in \widetilde{\Lambda}^{0,n}(J^\infty E)$, are called \textit{local}. Under suitable boundary conditions for~$\phi$'s the map $a \mapsto A$ def\/ines an isomorphism of vector spaces, which gives rise to a pulled-back Lie bracket on $\widetilde{\Lambda}^{0,n}(J^\infty E)$. This last bracket is determined by the symplectic structure
\begin{gather}\label{ops}
\omega= \delta \Phi_A^\ast\wedge \delta\Phi^A\wedge d^nx
\end{gather}
on the jet bundle $J^\infty E$ of f\/ields and antif\/ields. By def\/inition, $\operatorname{gh} (\omega)= - 1$ and $\epsilon (\omega)=1$. We will denote this Lie bracket by the same round brackets.

The central goal of the BV formalism is the construction of a \textit{master action}. This is given by a local functional
\begin{gather*}
 S[\phi]=\int_M (j^\infty \phi)^\ast(L)
\end{gather*}
obeying the \textit{classical master equation}
\begin{gather}\label{BVME}
(S,S)=0\quad \Leftrightarrow\quad \{L,L\}\simeq 0 .
\end{gather}
The master Lagrangian $L$ is required to be of ghost number zero and start with the Lagrangian $L_0$ of the original f\/ields to which one couples vertices involving antif\/ields. All these vertices can be found systematically from the master equation (\ref{BVME}) by means of the \textit{homological perturbation theory} \cite{HT}.

Since the canonical symplectic structure (\ref{ops}) of the BV formalism is nondegenerate, any form of top horizontal degree is Hamiltonian, i.e., $\Lambda_\omega^{n}(J^\infty E)=\Lambda^{0,n}(J^\infty E)$. Then the action of the classical BRST dif\/ferential on $\Lambda^{0,n}(J^\infty E)$ is canonically generated by the master Lagrangian:
\begin{gather}\label{CLBRSTD}
Q=(L ,\,\cdot\,) .
\end{gather}
Because of the classical master equation (\ref{BVME}), the operator $Q$ squares to zero. Fixing a vo\-lume form $d^nx$ on $M$ allows us to identify the spaces $\Lambda^{0,n}(J^\infty E)$ and $\Lambda^{0,0}(J^\infty E)$. Upon this identif\/ication the action (\ref{CLBRSTD}) induces that in the space of $0$-forms. The latter specif\/ies the evolutionary vector f\/ield $Q$ completely.

Thus, we see that the standard ingredients of the BV formalism -- the antibracket and the master action -- def\/ine a gauge system in our sense; in so doing, the classical BRST dif\/ferential is generated by the master action through the antibracket.

The following statement is of particular importance for the BV formalism.

\begin{Proposition}\label{41}
Let $(Q,\omega)$ be a gauge system, with $\omega$ being a symplectic form. Then a~Hamiltonian vector field $X_A$ with $\operatorname{gh} A\neq -1$ is a~symmetry iff it preserves the Hamiltonian of~$Q$, i.e., $\mathcal{L}_{X_A} L\simeq 0$.
\end{Proposition}
\begin{proof}
According to equations~(\ref{PB1}) and (\ref{LAHom}) we have
\begin{gather*}
\mathcal{L}_{X_A} L \simeq (-1)^{\epsilon(X_A)}\{A,L\}\qquad \mbox{and}\qquad i_{[X_A,Q]}\omega \simeq \delta \{A,L\} .
\end{gather*}
If $X_A$ preserves~$L$, then the r.h.s.\ of the second equation vanishes and we conclude that $[X_A,Q]\in \operatorname{ker} \omega$. For symplectic forms this implies that $[X_A,Q]=0$; and hence, $X_A$ is a symmetry. Conversely, if $X_A$ is a symmetry, then the l.h.s.\ of the second equation vanishes and we get $\delta \{A,L\}\simeq 0$, where $\operatorname{gh}\{A,L\}=\operatorname{gh} A+1\neq 0$. Since the relative $\delta$-cohomology is trivial in nonzero ghost number, we conclude that $\{A,L\}\simeq 0$. By the f\/irst equation, this means that $X_A$ preserves $L$.
\end{proof}

Let us now turn to specif\/ic gauge systems.

\subsection{Maxwell's electrodynamics} In the BV approach \cite{HT}, the free electromagnetic f\/ield on a $4$-dimensional space-time mani\-fold~$M$ is described by the gauge potential\footnote{For the sake of simplicity, we assume that the gauge potential def\/ines a connection in a \textit{ trivial } $U(1)$-bundle over $M$. This makes possible to identify the space of abelian connections with the space of $1$-forms.} $A\in \Lambda^1(M)$, the ghost f\/ield $C\in \Lambda^0(M)$ as well as their antif\/ields $A^\ast\in \Lambda^3(M)$ and $C^\ast\in \Lambda^4(M)$. The ghost number distribution reads
\begin{gather*}
\operatorname{gh} (C^\ast) =-2 , \qquad\operatorname{gh} (A^\ast) =-1 ,\qquad \operatorname{gh} (A) =0 ,\qquad \operatorname{gh} (C) =1 ,
\end{gather*}
and the Grassmann parity is just the ghost number modulo $2$. The space of f\/ields and antif\/ields is endowed with the canonical symplectic structure
\begin{gather}\label{wED}
\omega=\delta A\wedge \delta A^\ast+\delta C\wedge \delta C^\ast ,\qquad \operatorname{gh}(\omega)=-1 .
\end{gather}
The action of the classical BRST dif\/ferential is given by the equations
\begin{gather}\label{BRSTtrans}
\delta_Q C^\ast=dA^\ast ,\qquad \delta_Q A^\ast=d \tilde{F} ,\qquad \delta_Q A=dC ,\qquad \delta_Q C=0 .
\end{gather}
Here $F=dA$ is the strength of the electromagnetic f\/ield and $\tilde{F}=\ast F$ is its Hodge dual. Notice that the Maxwell equations $d\tilde{F}=0$ are the part of the def\/ining relations for the zero locus of~$Q$. The Hamiltonian of the classical BRST dif\/ferential $Q$ is given by the BV master Lagrangian
\begin{gather}\label{EDL}
L=\frac12 F\wedge \tilde{F}+A^\ast\wedge dC ,\qquad i_Q\omega \simeq \delta L .
\end{gather}
As a consequence of $Q^2=0$, the master Lagrangian $L$ satisf\/ies the BV master equation
\begin{gather*}
\{L,L\}=-i^2_Q\omega\simeq 0 .
\end{gather*}
Applying the BRST dif\/ferential to (\ref{wED}) yields the descendent presymplectic structure
\begin{gather*}
\omega_1= \delta C\wedge \delta A^\ast+\delta A\wedge \delta \tilde{F} ,\qquad \delta_Q\omega=d \omega_1 ,\qquad \operatorname{gh}(\omega_1)=0 .
\end{gather*}
The descendent Hamiltonian of $Q$ is given by the conserved BRST current $L'_1=2L_1$, where
\begin{gather*}
L_1=C\wedge d\tilde{F} ,\qquad dL_1=\delta_QL .
\end{gather*}
The current is obviously trivial as $L_1\approx 0$. Integrating $L'_1$ over a Cauchy hypersurface $N\subset M$, we get the classical BRST charge $\Omega=\int_N L'_1$. Again, in view of the equation $Q^2=0$, the BRST current obeys the classical master equation
\begin{gather*}
\{L'_1,L'_1\}_1=-i_Q^2\omega_1\simeq 0 .
\end{gather*}

Acting now by the BRST dif\/ferential on $\omega_1$, we get one more presymplectic structure of ghost number one
\begin{gather*}
\omega_2=\delta C\wedge \delta \tilde{F} ,\qquad \delta_Q\omega_1=d\omega_2 ,\qquad \operatorname{gh}(\omega_2)=1 .
\end{gather*}
This last form, being ``absolutely'' invariant under the BRST transformations (\ref{BRSTtrans}), leaves no further descendants. Thus, the length of Maxwell's electrodynamics relative to the BV symplectic structure (\ref{wED}) equals $2$.

Given a Killing vector $\xi$ of the background metric, one can def\/ine an even vector f\/ield $X$ on the space of f\/ields and antif\/ields. The latter is determined by the relations
\begin{gather}\label{EDsym}
\delta_X \Phi=L_\xi \Phi ,\qquad \Phi=(A,C,A^\ast, C^\ast) .
\end{gather}
Here $\delta_X=i_X \delta+\delta i_X$ denotes the Lie derivative along the evolutionary vector f\/ield $X$ on the jet space of f\/ields and antif\/ields, while $L_\xi=di_\xi+i_\xi d$ is the usual Lie derivative on horizontal forms. Using Proposition \ref{41}, one can easily see that $X$ is a symmetry of the gauge system, i.e., $[Q,X]=0$. Furthermore, this symmetry is Hamiltonian relative to (\ref{wED}):
\begin{gather*}
i_X\omega \simeq \delta \Xi ,
\end{gather*}
where
\begin{gather*}
\Xi= -A^\ast\wedge L_\xi A - C^\ast\wedge L_\xi C ,\qquad \operatorname{gh}(\Xi)=-1 .
\end{gather*}
The Hamiltonian $\Xi$ generates the symmetry transformations (\ref{EDsym}) through the BV bracket
\begin{gather*}
\delta_X\Phi=-\{\Xi, \Phi\} ,\qquad \Phi=(A,C, A^\ast, C^\ast) .
\end{gather*}

By Proposition \ref{27}, $\Xi$ is an observable. We have
\begin{gather*}
\delta_Q \Xi= d\Theta , \qquad \Theta \approx \frac12 \big(F\wedge i_\xi \tilde{F}-\tilde{F}\wedge i_\xi F\big) ,\qquad \operatorname{gh}(\Theta)=0 .
\end{gather*}
Thus, to each Killing vector we can associate a conserved current $\Theta$. Using the Hodge dualization, we can write
\begin{gather*}
 \ast\Theta=\xi^\mu T_{\mu\nu}dx^\nu ,
\end{gather*}
where $\{x^\nu\}$ are local coordinates on $M$ and $T_{\mu\nu}=T_{\nu\mu}$ is nothing but the energy-momentum tensor of the electromagnetic f\/ield. Since $\delta_Q \Theta =0$, the observable $\Xi$ gives no lower-degree conservation laws. In other words, the length of the space-time symmetry $X$ is equal to $1$.

The free electromagnetic f\/ield admits also a symmetry of length 2. This is generated by the evolutionary vector f\/ield $Y$ def\/ined by the relations
\begin{gather*}
\delta_Y C^\ast=0 ,\qquad \delta_Y A^\ast=0 ,\qquad \delta_Y A=0 ,\qquad \delta_Y C=1 .
\end{gather*}
One can easily check that $Y$ is a nontrivial Hamiltonian symmetry of the master Lagran\-gian~(\ref{EDL}), i.e.,
\begin{gather*}
 [Q,Y]=0 ,\qquad i_Y\omega =\delta C^\ast ,\qquad \mathcal{L}_Y L\simeq 0 .
\end{gather*}
The symmetry owes its origin to the \textit{global reducibility} \cite{BBH, T} of the gauge transformations $\delta_\varepsilon A=d\varepsilon$, meaning that we can shift the gauge parameter $\varepsilon$ by an arbitrary constant for no change of $\delta_\varepsilon A$. By Proposition \ref{27}, $C^\ast$ is an observable giving rise to the sequence of conservation laws
\begin{gather*}
\delta_Q C^\ast = d A^\ast , \qquad \delta_Q A^\ast =d \tilde{F} ,\qquad \delta_Q \tilde{F}=0 .
\end{gather*}
As is seen, the forms $A^\ast$ and $\tilde{F}$ def\/ine the conserved currents of degrees $3$ and $2$, so that the length of the symmetry $Y$ is 2. The latter conserved current has ghost number zero and expresses Gauss' law:
\begin{gather*}
q=\int_S \tilde{F} .
\end{gather*}
In words, it states that the net electric f\/lux through any closed, space-like surface $S$ is equal to the net electric charge $q$ within that closed surface.

\subsection{The abelian Chern--Simons theory} Consider now the Chern--Simons theory for a trivial $U(1)$-bundle over a $3$-dimensional mani\-fold~$M$. The theory is known to be purely gauge, possessing no local degrees of freedom. The spectrum of the BV f\/ields and antif\/ields is given by the gauge potential $A\in \Lambda^1(M)$, ghost f\/ield $C\in \Lambda^0(M)$ and their conjugate antif\/ields $A^\ast\in \Lambda^2(M)$ and $C^\ast\in \Lambda^3(M)$. These are prescribed the following ghost numbers:
\begin{gather*}
\operatorname{gh}(C^\ast)=-2 ,\qquad \operatorname{gh}(A^\ast)=-1 , \qquad\operatorname{gh}(A)=0 ,\qquad \operatorname{gh}(C)=1 .
\end{gather*}
The classical BRST dif\/ferential $Q$ acts in the space of f\/ields and antif\/ields according to the relations
\begin{gather}\label{CSD}
\delta_Q C^\ast=dA^\ast , \qquad \delta_Q A^\ast=dA ,\qquad \delta_Q A=dC , \qquad \delta_QC=0 .
\end{gather}
This action is Hamiltonian with respect to the canonical BV symplectic structure
\begin{gather}\label{CSSS}
\omega=\delta A\wedge \delta A^\ast + \delta C\wedge \delta C^\ast ,\qquad \operatorname{gh}(\omega)=-1 ,
\end{gather}
and the Hamiltonian for $Q$ is given by the ghost-extended Chern--Simons' Lagrangian
\begin{gather*}
L=\frac12 A\wedge dA +dC\wedge A^\ast ,\qquad i_Q\omega \simeq \delta L .
\end{gather*}
As usual the BRST dif\/ferential of $L$ gives rise to the conserved BRST current $L'_1=2L_1$, where
\begin{gather*}
 L_1=\frac 32 C\wedge dA ,\qquad \qquad dL_1=\delta_Q L .
\end{gather*}
This current is necessarily trivial. Starting from the BV symplectic structure (\ref{CSSS}), one can def\/ine the full sequence of descendent presymplectic structures of increasing ghost number:
\begin{alignat*}{4}
& \delta_Q \omega =d \omega_1 ,\qquad&& \omega_1=\tfrac 12 \delta A\wedge \delta A+\delta C\wedge \delta A^\ast ,\qquad && \operatorname{gh}(\omega_1)=0 ,& \\
& \delta_Q \omega_1 =d \omega_2 , \qquad&& \omega_2=\delta C\wedge \delta A , \qquad &&\operatorname{gh}(\omega_2)=1 ,&\\
 &\delta_Q \omega_2 =d \omega_3 ,\qquad &&\omega_3=\tfrac12\delta C\wedge \delta C ,\qquad &&\operatorname{gh} (\omega_3)=2 .&
\end{alignat*}
In particular, the BRST current $L'_1$ obeys the master equation
\begin{gather*}
\{L'_1,L'_1\}_1\simeq 0
\end{gather*}
relative to the Lie bracket associated with $\omega_1$.

Notice that the gauge symmetry transformations for the Chern--Simons f\/ield, being identical in form to those of the electromagnetic f\/ield, are globally reducible and this leads to the odd symmetry $Y$.
Explicitly,
\begin{gather*}
\delta_Y C^\ast=0 , \qquad \delta_Y A^\ast=0 ,\qquad \delta_Y A=0 , \qquad \delta_YC=1 .
\end{gather*}
The symmetry is obviously Hamiltonian,
\begin{gather*}
i_Y\omega = \delta C^\ast ,
\end{gather*}
and the Hamiltonian $C^\ast$ gives rise to the conserved currents $A^\ast$, $A$, and $C$ as is seen from (\ref{CSD}). Computing the various Lie brackets of the currents, we f\/ind
\begin{gather*}
\{A^\ast,A^\ast\}_1=0 ,\qquad \{A,A\}_2=0 ,\qquad \{C,C\}_3=-1 .
\end{gather*}
Here we face with the phenomenon of central extension mentioned at the end of Section~\ref{section4}. Namely, the abelian super-Lie algebra $[Y,Y]=0$ of
symmetry gets central extension when evaluated at the level of conserved currents.

The integral of the conserved current $A$ over a loop $\gamma \subset M$ gives the conserved ``charge''
\begin{gather*}
h=\int_\gamma A ,
\end{gather*}
which is nothing but the holonomy of the f\/lat abelian connection $A$. One can think of these holonomy invariants as global degrees of freedom in the Chern--Simons theory.

\subsection{Linearized gravity} Our last example is the free massless f\/ield of spin 2. This theory can be obtained by linearizing Einstein's equations about the f\/lat background. In the vierbein formalism, the background geometry is described by the vierbein $\{h^a\}$, which is assumed to be given by a set of four linearly independent, closed $1$-forms on the Minkowski space. The small f\/luctuations of geometry around the f\/lat background are described by the collection of ten 1-form f\/ields $e^a$ and $\omega^{ab}=-\omega^{ba}$. These are identif\/ied, respectively, with the perturbations of the f\/lat vierbein and spin connection. In accordance with the general prescriptions of the BV formalism, this set of f\/ields is extended by the ghost f\/ields $c^a$ and $c^{ab}=-c^{ba}$, associated with the general coordinate and local Lorentz invariance, as well as the antif\/ields $e^\ast_a$, $\omega^\ast_{ab}$, $c^\ast_a$, and $c^\ast_{ab}$. The form degrees and the ghost numbers of the introduced
f\/ields are collected in the following table:
\begin{center}
\begin{tabular}{|c|c|c|c|c|c|c|c|c|}
 \hline
 & $c^\ast_a$& $c^\ast_{ab}$ & $e^\ast_a$ &$\omega^\ast_{ab}$ & $e^a$ &$\omega^{ab}$ &$c^a$& $c^{ab}$ \\ \hline
 $\deg$ & 4 & 4 & 3 & 3 & 1 &1 & 0 & 0\\ \hline
 $\operatorname{gh}$ & $-2$ & $-2$ & $-1$ &$-1$ & 0 &0 & 1& 1\\ \hline
\end{tabular}
\end{center}
The canonical BV symplectic structure assumes the form
\begin{gather}\label{gr-pr-str}
\omega=\delta e^a\wedge \delta e^\ast_a+\delta \omega^{ab}\wedge\delta \omega^\ast_{ab}+\delta c^a\wedge \delta c^\ast_a +\delta
c^{ab}\wedge\delta c_{ab}^\ast .
\end{gather}
In order to def\/ine the classical BRST dif\/ferential $Q$ it is convenient to introduce the following collections of background $1$- and $2$-forms:
\begin{gather*}h_{abc}=\varepsilon_{abcd}h^d ,\qquad H_{ab}=h^c\wedge h_{abc} ,\end{gather*}
with $\varepsilon_{abcd}$ being the Levi-Civita symbol. Then the action of $Q$ is given by the relations
\begin{alignat*}{3}
&\delta_Q c^\ast_a=de^\ast_a ,\qquad && \delta_Q c^\ast_{ab}=d\omega^\ast_{ab}-\tfrac12\big(e^\ast_a\wedge h_b-e^\ast_b\wedge h_a\big) ,& \\
& \delta_Q e^\ast_a=d\omega^{bc}\wedge h_{abc} , \qquad && \delta_Q \omega^\ast_{ab}=de^c\wedge h_{cab}+\tfrac12\big(\omega_b{}^c\wedge H_{ac}-\omega_a{}^c\wedge H_{bc}\big) ,& \\
 &\delta_Qe^a=dc^a+c^{ab}h_b ,\qquad && \delta_Q \omega^{ab}=dc^{ab} ,& \\
& \delta_Q c^a=0 ,\qquad &&\delta_Q c^{ab}=0 . &
\end{alignat*}
This action is Hamiltonian relative to (\ref{gr-pr-str}) and is generated by the BV master Lagrangian
\begin{gather*}
L=e^a\wedge d\omega^{bc}\wedge h_{abc}+\tfrac12 \omega^a{}_{c}\wedge \omega^{cb}\wedge H_{ab}+e^\ast_a\wedge \big(dc^a+c^{ab}h_b\big)+\omega^\ast_{ab}\wedge dc^{ab} .
\end{gather*}
Hereafter all indices are raised and lowered by means of the Minkowski metric. The BV symplectic structure (\ref{gr-pr-str}) gives rise to the following sequence of presymplectic forms of decreasing horizontal degree and increasing ghost number:
\begin{alignat*}{4}
& \delta_Q\omega=d\omega_1 ,\qquad && \omega_1=\delta e^a\wedge \delta \omega^{bc}\wedge h_{abc}+\delta c^a\wedge \delta e^\ast_a+\delta c^{ab}\wedge \delta \omega^\ast_{ab} ,\qquad && \operatorname{gh}(\omega_1)=0 ,& \\
& \delta_Q\omega=d\omega_2 ,\qquad && \omega_2=\big(\delta c^a\wedge \delta\omega^{bc}+\delta c^{ab}\wedge\delta e^c\big)\wedge h_{abc} ,\qquad&& \operatorname{gh} (\omega_2)=1 , &\\
& \delta_Q\omega=d\omega_3 ,\qquad && \omega_3=\delta c^a\wedge \delta c^{bc}\wedge h_{abc} , \qquad && \operatorname{gh}(\omega_3)=2 ,&
\end{alignat*}
and $\delta_Q\omega_3=0$. Thus, the length of the linearized gravity is $3$.

As with the Maxwell electrodynamics, the isometries of the Minkowski space give rise to the conserved energy-momentum tensor of spin-2 f\/ield. This conservation law, however, does not survive in the full nonlinear theory. In general relativity, the canonical energy-momentum tensor is known to vanish on shell. Much more interesting are the lower-degree conservation laws that are present in the theory. These can be constructed as follows.

Let $\xi^a$ and $\xi^{ab}=-\xi^{ba}$ be some functions on the Minkowski space. Def\/ine the odd evolutionary vector f\/ield $Y$ by the relations
\begin{gather*}
\delta_Y c^a=\xi^a ,\qquad \delta_Y c^{ab}=\xi^{ab} ,\qquad \delta_Y (\text{the other f\/ields})=0 .
\end{gather*}
As above, by $\delta_Y$ we denoted the Lie derivative along $Y$. Using Proposition \ref{41}, one can see that the vector f\/ield $Y$ is a symmetry if\/f the following equations are satisf\/ied:
\begin{gather*}
d\xi^a=\xi^{ab}h_b ,\qquad d\xi^{ab}=0 .
\end{gather*}
The general solution to these equations is obvious. If we choose $h^a=dx^a$, where $x^a$ are the Cartesian coordinates on $\mathbb{R}^{1,3}$, then
\begin{gather*}
\xi^{a}(x)=\zeta^a+\zeta^{ab}x_b ,\qquad \xi^{ab}(x)=\zeta^{ab}
\end{gather*}
for arbitrary constant parameters $\zeta^a$ and $\zeta^{ab}=-\zeta^{ba}$. The $\xi$'s are naturally identif\/ied with the ten Killing vectors of the Minkowski metric.
The symmetry $Y$ is clearly Hamiltonian:
\begin{gather*}
i_Y\omega=\delta H ,\qquad H=\xi ^a c^\ast_a +\xi^{ab}c_{ab}^\ast .
\end{gather*}
The Hamiltonian $H$, being a physical observable, generates the following sequence of conserved currents:
\begin{alignat*}{5}
& \delta_Q H=dJ_1 ,\qquad && J_1=\xi^ae^\ast_a+\xi^{ab}\omega^\ast_{ab} ,\qquad&& \deg J_1=3 ,\qquad && \operatorname{gh} J_1=-1 ,& \\
& \delta_Q J_1=dJ_2 ,\qquad && J_2=\big(\xi^a\omega^{bc}+\xi^{ab}e^c\big)\wedge h_{abc} ,\qquad && \deg J_2=2 ,\qquad && \operatorname{gh} J_2=0 ,& \\
& \delta_Q J_2=dJ_3 ,\qquad && J_3=\big(\xi^a c^{bc}+\xi^{ab}c^c\big)\wedge h_{abc} ,\qquad && \deg J_3= 1 ,\qquad && \operatorname{gh} J_3=1 ,&
\end{alignat*}
and $\delta_Q J_3=0$.
The integral of the $10$-parameter family of conserved currents $J_2$ over a closed, space-like surface $S\subset \mathbb{R}^{1,3}$ gives the net energy-momentum~$\mathcal{P}$ and the angular momentum $\mathcal{M}$ of the spin-2 f\/ield produced by the sources inside $S$:
\begin{gather*}
\mathcal{P}_a=\int_S \omega^{bc}\wedge h_{abc} ,\qquad \mathcal{M}_{ab}=\int_S e^c\wedge h_{abc}+\frac12\big(x_b\omega^{dc}\wedge h_{adc} -x_a\omega^{dc}\wedge h_{bdc}\big) .
\end{gather*}

Although these conserved currents do not extend into the full nonlinear theory of gravity, they can be used for the derivation of asymptotic conservation laws (e.g., ADM energy) in general relativity via a surface integral at inf\/inity \cite{AT,BB,T}.

Evaluating now the descendent brackets of the conserved currents above, one can easily f\/ind
\begin{gather*}
\big\{J_1^\xi,J_1^{\xi'}\big\}_1=0 ,\qquad \big\{J_2^\xi,J_2^{\xi'}\big\}_2=0 ,\qquad \big\{J_3^\xi,J_3^{\xi'}\big\}_3=-\big(\xi^a\xi'^{bc}+\xi'^{a}\xi^{bc}\big)h_{abc} ,
\end{gather*}
As with the Chern--Simons theory, the Lie brackets of the zero-degree currents get the central extension.

\appendix

\section{Jet bundles and the variational bicomplex}\label{A}

In this appendix, we brief\/ly recall some basic elements from the theory of jet bundles and variational bicomplex, which are relevant for our discussion. A more systematic exposition of these concepts can be found in \cite{Anderson,Dickey,Olver,Saunders}.

The starting point of any f\/ield theory is a locally trivial f\/iber bundle $\pi\colon E\rightarrow M$ whose base is identif\/ied with the space-time manifold and which sections are called \textit{classical fields}. For the sake of simplicity, we restrict ourselves to f\/ields with values in vector bundles, although the subsequent discussion could be straightforwardly extended to general smooth bundles. On the other hand, to accommodate bosonic and fermionic f\/ields, we allow the f\/ibers of $E$ to be superspaces with a given number of even and odd dimensions; in so doing, the base $M$ remains
a pure even manifold. The Grassmann parity of a homogeneous object $A$ is denoted by $\epsilon(A)\in \{0,1\}$.

Associated with a vector bundle $\pi\colon E\rightarrow M$ is the vector bundle $\pi_k\colon J^kE \rightarrow M$ of $k$-jets of sections of $E$. By def\/inition, the $k$-jet $j^k_x\phi$ at $x\in M$ is just the equivalence class of the section $\phi\in \Gamma(E)$, where two sections are considered to be equivalent if they have the same Taylor development of order $k$ at $x\in M$ in some (and hence any) adapted coordinate chart. It follows from the def\/inition that each section $\phi$ of $E$ induces the section $j^k\phi$ of $J^kE$ by the rule $(j^k\phi)(x)=j^k_x\phi$. The latter is called the $k$-jet prolongation of $\phi$.

If $E|_U\simeq\mathbb{R}^m\times U$ is an adapted coordinate chart with local coordinates $(x^i, \phi^a)$, then $(x^i, \phi^a, \phi^a_i,\ldots, \phi^a_{i_1\cdots i_k})$ are local coordinates in $J^kE$ and the induced section $j^k\phi$ is given in these coordinates by
\begin{gather*}
x \mapsto \big(x, \phi^a(x), \partial_i\phi^a(x),\ldots, \partial_{i_1}\cdots\partial_{i_k}\phi^a(x)\big) .
\end{gather*}
We use the multi-index notation and the summation convention through the paper. A multi-index $I=i_1i_2\cdots i_n$ represents the corresponding set of symmetric covariant indices. The order of the multi-index is given by $|I|=k$. By def\/inition we set $Ij=jI=i_1i_2\cdots i_kj$. With the multi-index notation we can write the partial derivatives of f\/ields as $\partial_{i_1}\cdots \partial_{i_k}\phi^a=\partial_I \phi^a$ and the set of local coordinates on $J^kE|_U$ as $(x^i,\phi^a_I)$, $|i|\leq k$.

Jet bundles come with natural projection $J^kE\rightarrow J^{k-1}E$ def\/ined by forgetting all the derivatives of order $k$. One can easily see that this projection gives $J^kE$ the structure of an af\/f\/ine bundle over the base $J^{k-1}E$. Thus, we have the inf\/inite sequence of surjective submersions
\begin{gather}\label{i-lim}
\xymatrix{\cdots\ar[r]& J^3E\ar[r]&J^2E\ar[r]&J^1E\ar[r]&J^0E\simeq E} .
\end{gather}
The inf\/inite order jet bundle $J^\infty E$ is now def\/ined as the inverse limit over the jet order $k$:
\begin{gather*}
J^\infty E=\lim_{\longleftarrow}J^k E .
\end{gather*}
Let $\Lambda^\ast(J^kE)$ denote the space of dif\/ferential forms on $J^kE$. The sequence of projections (\ref{i-lim}) gives rise to the chain of pullback maps
\begin{gather*}
\xymatrix{\cdots&\ar[l]\Lambda^\ast(J^3E)&\ar[l]\Lambda^\ast (J^2E)&\ar[l]\Lambda^\ast(J^1E)&\ar[l]\Lambda^\ast(J^0E)} .
\end{gather*}
This allows one to def\/ine the space of dif\/ferential forms on $\Lambda(J^\infty E)$ as the direct limit
\begin{gather*}
\Lambda^\ast(J^\infty E)=\lim_{\longrightarrow} \Lambda^\ast(J^kE) .
\end{gather*}
 According to this def\/inition each dif\/ferential form on $J^\infty E$ is the pullback of a smooth form on some f\/inite jet bundle $ J^kE$. As usual, the smooth functions on $J^\infty E$ are identif\/ied with the $0$-forms. For notational simplicity, we will not distinguish between a form on $J^\infty E$ and its representatives in f\/inite-dimensional jet bundles. The exterior dif\/ferential on $\Lambda^\ast (J^\infty E)$ will be denoted by $D$.

The de Rham complex $(\Lambda^\ast (J^\infty E), D)$ of dif\/ferential forms on $J^\infty E$ possesses the dif\/ferential ideal $\mathcal{C}(J^\infty E)$ of contact forms. By def\/inition, $\alpha\in \mathcal{C}(J^\infty E)$ if\/f $(j^\infty \phi)^\ast \alpha=0$ for all sections $\phi\in \Gamma(E)$. The ideal $\mathcal{C}(J^\infty E)$ is known to be generated by the contact $1$-forms, which in local coordinates take the form $\delta \phi^a_I:=D\phi^a_I-\phi^a_{Ij}Dx^j$. Using the contact forms, one can split the exterior dif\/ferential $D$ into the sum of \textit{horizontal} and \textit{vertical differentials}, namely, $D=d+\delta$ where
\begin{gather*}
d=dx^j\wedge \left( \frac{\partial}{\partial x^j}+\phi^a_{Ij}\frac{\partial}{\partial \phi^a_{I}}\right) ,\qquad \delta=\delta \phi_I^a\wedge \frac{\partial_l}{\partial \phi^a_I} .
\end{gather*}
It is easy to see that
\begin{gather*}
d^2=0 ,\qquad \delta^2=0 ,\qquad d\delta+\delta d=0 .
\end{gather*}
Any $p$-form of $\Lambda^p(J^\infty E)$ can now be written as a f\/inite sum of homogeneous forms
\begin{gather*}
f dx^{i_1}\wedge \cdots\wedge dx^{i_r}\wedge \delta \phi^{a_1}_{I_1}\wedge\cdots\wedge \delta\phi^{a_s}_{I_s}
\end{gather*}
of horizontal degree $r$ and vertical degree $s$, with $r+s=p$ and~$f$ being a smooth function on~$J^\infty E$. The \textit{variational bicomplex} is the double complex $(\Lambda^{\ast,\ast}(J^\infty E), \delta, d)$ of dif\/ferential forms on~$J^\infty E$:
\begin{gather*}
\xymatrix{
&&\vdots &\vdots&&\vdots\\
&0 \ar[r] & \Lambda^{2,0}(J^\infty E) \ar[u]^\delta \ar[r]^-{d} & \Lambda^{2,1}(J^\infty E) \ar[r]^-{d} \ar[u]^\delta & \ldots \ar[r]^-{d} & \Lambda^{2,n}(J^\infty E) \ar[u]^\delta \\
&0 \ar[r] & \Lambda^{1,0}(J^\infty E) \ar[u]^\delta \ar[r]^-{d} & \Lambda^{1,1}(J^\infty E) \ar[r]^-{d} \ar[u]^\delta & \ldots \ar[r]^-{d} & \Lambda^{1,n}(J^\infty E) \ar[u]^\delta \\
& 0\ar[r]& \Lambda^{0,0}(J^\infty E) \ar[u]^\delta \ar[r]^-{d} & \Lambda^{0,1}(J^\infty E) \ar[r]^-{d} \ar[u]^\delta & \ldots \ar[r]^-{d} & \Lambda^{0,n}(J^\infty E) \ar[u]^\delta. \\
}\end{gather*}
The most important result concerning the variational bicomplex for the vector bundle $\pi\colon E\rightarrow M$ is that all the columns and interior rows of the diagram above are exact.

It is possible to augment the variational bicomplex from below by the de Rham complex of the base manifold:
\begin{gather*}
\xymatrix{
 0\ar[r]&\Lambda^{0,0}(J^\infty E) \ar[r]^-{d} & \Lambda^{0,1}(J^\infty E) \ar[r]^-{d} & \ldots \ar[r]^-{d} & \Lambda^{0,n}(J^\infty E) \\
 0\ar[r]& \Lambda^{0}(M) \ar[u]^{\pi_\infty^\ast} \ar[r]^-{d} & \Lambda^{1}(M) \ar[r]^-{d} \ar[u]^{\pi_\infty^\ast} & \ldots \ar[r]^-{d} & \Lambda^{n}(M) \ar[u]^{\pi^\ast_\infty}\\
 & 0 \ar[u] & 0 \ar[u] & & 0. \ar[u]
}\end{gather*}
The resulting bicomplex remains exact in columns.

As with any bicomplex, one can consider the relative cohomology of ``$\delta$ modulo $d$''. It is described by the groups $H^{p,q}(J^\infty E; \delta/d)$ which are essentially the cohomology groups of the quotient complex $\widetilde{\Lambda}^{p,q}(J^\infty E)=\Lambda^{p,q}(J^\infty E) /d\Lambda^{p,q-1}(J^\infty E)$ with dif\/ferential induced by $\delta$. In the main text, we often use the following statement about the relative $\delta$-cohomology.

\begin{Proposition}[{\cite[Section~19.3.9]{Dickey}}]\label{Prop-A}\label{A1}
\begin{gather*}
H^{p,q} (J^\infty E; \delta/d)=0\qquad \mbox{for}\quad p > 0\qquad \mbox{and}\qquad H^{0,q}(J^\infty E; \delta/d)\simeq \Lambda^q(M)/d\Lambda^{q-1}(M) .
\end{gather*}
\end{Proposition}

The quotient $\delta$-complex $\widetilde{\Lambda}^{p,n}(J^\infty E)={\Lambda}^{p,n}(J^\infty E)/d{\Lambda}^{p,n-1}(J^\infty E)$ provides a natural augmentation of the variational bicomplex from the right:
\begin{gather*}
\xymatrix{&\vdots&\vdots &\\
 \ar[r]^-{d}&\Lambda^{2,n}(J^\infty E)\ar[u]^\delta \ar[r]^-{p} & \widetilde{\Lambda}^{n,2}(J^\infty E) \ar[r] \ar[u]^{\delta} & 0 \\
\ar[r]^-{d}& \Lambda^{1,n}(J^\infty E) \ar[u]^{\delta} \ar[r]^-{p} & \widetilde{\Lambda}^{1,n}(J^\infty E) \ar[r] \ar[u]^{\delta} & 0\\
\ar[r]^-{d}& \Lambda^{0,n}(J^\infty E) \ar[u]^{\delta} \ar[r]^-{p} & \widetilde{\Lambda}^{0,n}(J^\infty E) \ar[r] \ar[u]^{\delta} & 0.\\
}\end{gather*}
 $p$ being the canonical projection onto the quotient space. Proposition \ref{Prop-A} ensures that the appended column is exact. The space $\widetilde{\Lambda}^{0,n}(J^\infty E)$ is usually identif\/ied with the space of local functionals of f\/ields. The correspondence between the two spaces is established by the assignment
 \begin{gather*}
 \widetilde{\Lambda}^{0,n}(J^\infty E) \ni [a] \mapsto A[\phi]=\int_M (j^{\infty}\phi)^\ast (a) ,
 \end{gather*}
 with $\phi$ being a compactly supported section of $E$.

The space $\Lambda^{1,n}(J^\infty E)$ has a distinguished subspace spanned by the \textit{source forms}. These are given by f\/inite sums of the forms
\begin{gather*}
\alpha\wedge \delta \phi^a ,
\end{gather*}
where $\alpha \in \Lambda^{0,n}(J^\infty E)$. Using the exactness of the variational bicomplex one can prove the following

\begin{Proposition}[\cite{Dickey}]\label{SF}
For any $(1,n)$-form $\alpha$ there exists a unique source form $\beta$ and a~$(1,n-1)$-form $\gamma$ such that
\begin{gather*}
\alpha=\beta+d\gamma .
\end{gather*}
The form $\gamma$ is uniquely determined up to a $d$-exact form. In particular, a nonzero source form can never be $d$-exact.
\end{Proposition}

Given $\lambda\in \Lambda^{0,n}(J^\ast E)$, we can apply the proposition above to $\delta \lambda$. We get
\begin{gather*}
\delta \lambda = \delta \phi^a\wedge \frac{\delta \lambda}{\delta \phi^a}+ d\gamma .
\end{gather*}
The coef\/f\/icients $\delta \lambda/\delta \phi^a$ def\/ining the source form are called the \textit{Euler-Lagrange derivative} of the form $\lambda$. Explicitly,
\begin{gather*}
\frac{\delta \lambda}{\delta \phi^a}=(-\partial)_I\frac{\partial \lambda }{\partial \phi^a_I} ,
\end{gather*}
where
\begin{gather}\label{dpat}
(-\partial)_I=(-1)^{|I|}\partial_I ,\qquad \partial_I=\partial_{i_1}\cdots\partial_{i_k} , \qquad \partial_i=\frac{\partial}{\partial x^i}+\phi^a_{Ii}\frac{\partial_l}{\partial\phi^a_{I}} .
\end{gather}

Dual to the space of $1$-forms on $J^\infty E$ is the space of vector f\/ields ${\mathfrak X}(J^\infty E)$. In terms of local coordinates, the elements of ${\mathfrak X}(J^\infty E)$ are given by the formal series
\begin{gather}\label{X}
X=X^i\frac{\partial}{\partial x^i} + X_{I}^a\frac{\partial_l}{\partial \phi^a_I} ,
\end{gather}
where $X^i$ and $X_I^a$ are smooth functions on $J^\infty E$. A vector f\/ield $X$ is called \textit{vertical} if $X^i=0$.

The operation $i_X$ of contraction of the vector f\/ield (\ref{X}) with a dif\/ferential form is def\/ined as usual: $i_X$ is a dif\/ferentiation of the exterior algebra $\Lambda^{\ast}(J^\infty E)$ of form degree $-1$ and the Grassmann parity $\widetilde{X}+1$ which action on the basis 1-forms is given by
\begin{gather*}
i_{X}\delta\phi^a_I=X^a_I ,\qquad i_{X}dx^j=X^j .
\end{gather*}

The operator of the Lie derivative along the vector f\/ield $X$ is def\/ined by the magic Cartan's formula
\begin{gather}\label{CMF}
L_X=i_XD+(-1)^{\epsilon(X)}Di_X .
\end{gather}
A vertical vector f\/ield $X$ is called \textit{evolutionary} if
\begin{gather*}
i_Xd+(-1)^{\epsilon(X)}di_X=0 .
\end{gather*}
It follows from the def\/inition that the vector f\/ield (\ref{X}) is evolutionary if\/f $X^i=0$ and $X^a_I=\partial_I (X^a)$, where $\partial_I$ is def\/ined by (\ref{dpat}). Hence, any vertical f\/ield of the form $X_0=X^a{\partial}/{\partial \phi^a}$ admits a~unique prolongation to an evolutionary vector f\/ield. We call $X_0$ the \textit{source vector field} for the evolutionary vector f\/ield $X$. (Our nomenclature is not standard; most of the authors prefer to call the vector f\/ield $X_0$ evolutionary, rather than its prolongation~$X$.) Note that the Lie derivative along the evolutionary vector f\/ield $X$ can be written as $L_X=i_X \delta+(-1)^{\epsilon({X})}\delta i_X$. It is clear that the evolutionary vector f\/ields form a closed Lie algebra.

\subsection*{Acknowledgements}
The work was partially supported by the RFBR grant No.~16-02-00284~A.

\pdfbookmark[1]{References}{ref}
\LastPageEnding

\end{document}